\documentclass[aip,reprint]{revtex4-1}

\draft 
\preprint{}
\usepackage[utf8]{inputenc}
\usepackage{amsmath,amsthm,amsfonts,amssymb}
\usepackage{graphicx}
\usepackage{natbib}
\usepackage{xcolor}
\usepackage{hyperref}
\hypersetup{
    colorlinks=true,       
    linktoc=page,
    linkcolor=blue,          
    citecolor=magenta,        
    filecolor=magenta,      
    urlcolor=blue,           
    linkbordercolor={1 1 1},
    citebordercolor={0 1 0},
    urlbordercolor={0 1 1}
}

\newtheorem{theorem}{Theorem}[section]
\theoremstyle{plain}

\theoremstyle{plain}
\newtheorem{example}{Example}[section]
\theoremstyle{definition}
\newtheorem{remark}{Remark}[section]
\theoremstyle{remark}

\newcommand{\hamil}{\mathcal{H}}
\newcommand{\Hamil}{\widetilde{\hamil}}

\begin{document}


\title{On Pseudo-Hermitian Hamiltonians}


\author{Soumendu Sundar Mukherjee}
\email[]{soumendu041@gmail.com}
\affiliation{Master of Statistics student, Indian Statistical Institute, Kolkata 700108, India}

\author{Pinaki Roy}
\email[]{pinaki@isical.ac.in}
\affiliation{Physics \& Applied Mathematics Unit, 
Indian Statistical Institute,
Kolkata 700108, India}
\date{\today}

\begin{abstract}
We investigate some questions on the construction of $\eta$ operators for pseudo-Hermitian Hamiltonians. We give a sufficient condition which can be exploited to systematically generate a sequence of $\eta$ operators starting from a known one, thereby proving the non-uniqueness of $\eta$ for a particular pseudo-Hermitian Hamiltonian. We also study perturbed Hamiltonians for which $\eta$'s corresponding to the original Hamiltonian still work.
\end{abstract}

\maketitle

\section{Introduction}
  In recent years there has been a growing interest in the study of non-Hermitian quantum mechanics \citep{bender1998real}, primarily because a class of these Hamiltonians admit real eigenvalues despite being non-Hermitian. Among the different non-Hermitian systems the ones with $\mathcal{PT}$-symmetry \citep{bender1999pt,bender2002complex,fernandez1998strong,levai2000systematic,levai2002pt} and those which are $\eta$-pseudo-Hermitian \citep{mostafazadeh2002pseudo1, mostafazadeh2002pseudo2, mostafazadeh2002pseudo3, mostafazadeh2002pseudo4, mostafazadeh2002pseudo5} have been the most widely studied. Lately non-Hermitian quantum mechanics have found applications in diverse areas of physics \citep{longhi2009bloch,kottos2010optical,longhi2010pt,makris2008beam,guo2009observation,ruter2010observation}. 

  Here our objective is to examine some features of $\eta$-pseudo-Hermitian models \citep{mostafazadeh2002pseudo1, mostafazadeh2002pseudo2, mostafazadeh2002pseudo3, mostafazadeh2002pseudo4, mostafazadeh2002pseudo5}. It may be recalled that a non-Hermitian Hamiltonian $\hamil$ is called $\eta$-pseudo-Hermitian if it satisfies the condition
  \begin{equation}\label{eq:pseudo}
  \eta \hamil \eta^{-1} = \hamil^\dag,
  \end{equation}
  where $\eta$ is a Hermitian linear automorphism \citep{mostafazadeh2002pseudo1, mostafazadeh2002pseudo2, mostafazadeh2002pseudo3, mostafazadeh2002pseudo4, mostafazadeh2002pseudo5}. However, there is no definite method to determine $\eta$ and for different models it has to be constructed in different ways. Another characteristic of $\eta$ is that it is not unique. Here we propose to examine the construction of $\eta$ operators for a class of $\eta$-pseudo-Hermitian matrix Hamiltonians. We shall also discuss briefly non-matrix Hamiltonians.

\section{Two Theorems on $\eta$}
  In this section we consider a particular class of pseudo-Hermitian Hamiltonians and present a general procedure to get an infinite set of $\eta$'s provided that we know one for a particular Hamiltonian. This will, in particular, explicitly demonstrate that $\eta$ is not unique. We shall also construct a class of non-Hermitian Hamiltonians whose members are pseudo-Hermitian with respect to the same $\eta$.

\begin{theorem}\label{t1}
  Let $\hamil$ be a non-Hermitian Hamiltonian that is pseudo-Hermitian, i.e. there exists a Hermitian linear automorphism $\eta$ such that $\eta \hamil \eta^{-1} = \hamil^\dag$. Suppose that $\hamil^\dag$ is invertible. Then there exists a possibly infinite no.of $\eta$'s, with respect to which $\hamil$ is pseudo-Hermitian.
\end{theorem}

\begin{proof}
  Let $\eta_0 = \eta$. Define a sequence $\{ \eta_k \}_{k \geq 1}$ by $\eta_k = \hamil^\dag \eta_{k-1}$. We claim that each $\eta$ in $\{ \eta_k \}_{k \geq 0}$ renders $\hamil$ pseudo-Hermitian.\\

  We prove the claim by induction. Suppose that for some $m \geq 0$, $\eta_m$ satisfies the claim, i.e., $\eta_m$ is a Hermitian linear automorphism with $\eta_m \hamil \eta_m^{-1} = \hamil^\dag$. Consider $\eta_{m+1}$. Clearly,
  \begin{align*}
  \eta_{m+1}^\dag & = (\hamil^\dag \eta_m)^\dag = \eta_m^\dag \hamil = \hamil^\dag \eta_m = \eta_{m+1}.
  \end{align*}
    
  Also, as both $\hamil^\dag$ and $\eta_m$ are linear automorphisms, so is $\eta_{m+1}$. It remains therefore to check that $\eta_{m+1} \hamil \eta_{m+1}^{-1}=\hamil^\dag$. Indeed we have
  \begin{align*}
    \eta_{m+1} \hamil \eta_{m+1}^{-1} &= \hamil^\dag \eta_m \hamil (\hamil^\dag \eta_m)^{-1} \\
    &= \hamil^\dag \eta_m \hamil \eta_m^{-1} (\hamil^\dag)^{-1}\\
    & = \hamil^\dag \hamil^\dag (\hamil^\dag)^{-1} \\
    &= \hamil^\dag.
  \end{align*}
  This completes the induction since $\eta_0$ satisfies the claim by definition.
\end{proof}

\begin{remark}
  By the definition in the above theorem we have $\eta_k = (\hamil^\dag)^k \eta$.
\end{remark}

\begin{example}\label{exm:berry}
  Consider now a quantum particle on a coordinate axis consisting of two points. The Hamiltonian for this system is given by a $2\times 2$ matrix of the form \citep{bender2002generalized}
\begin{equation}
  \hamil =
  \begin{pmatrix}
  	x & y\\
  	\bar{y} & \bar{x}
  \end{pmatrix} 
\end{equation}
where $\Im(x) \neq 0$. This Hamiltonian is $\mathcal{PT}$-symmetric \citep{bender2002generalized}. However it is also pseudo Hermitian w.r.t. the following choice of $\eta$:\\
  \[
    \eta =
    \begin{cases}
      \begin{pmatrix}
	0 & y\\
	\bar{y} & 0
      \end{pmatrix},	&\text{ if $y \neq 0$;}\\
      \begin{pmatrix}
	0 & 1\\
	1 & 0
      \end{pmatrix},	&\text{ otherwise.}
    \end{cases}
  \]
  Therefore, in the case $y = 0$, we have
  \[
  \eta_k = (\hamil^\dag)^k \eta =
  \begin{pmatrix}
    0 & \bar{x}^k\\
    x^k & 0
  \end{pmatrix}.
  \]
  Although the  operator can be also be obtained when $y \neq 0$, the expression for $(\hamil^\dag)^k$ is cumbersome and so, we omit it.
\end{example}

\begin{example}\label{exm:mostafa}
  Consider now a two level system associated with the classical motion of a simple harmonic oscillator with frequency $\omega(t)$. Then the relevant Hamiltonian is of the form \citep{mostafazadeh2002pseudo1,mostafazadeh2002pseudo2,mostafazadeh2002pseudo3,mostafazadeh2002pseudo4,mostafazadeh2002pseudo5}
\begin{equation}
  \hamil(t) =
  \begin{pmatrix}
 	 0 & i\\
     -i \omega(t)^2 & 0
  \end{pmatrix}.
\end{equation}
  It can be easily verified that in this case one can choose
  \[
    \eta =
    \begin{pmatrix}
    0 & i\\
    -i & 0
    \end{pmatrix}
  \]
   Here $|\hamil(t)| = -\omega(t)^2$, which is non-vanishing if $\omega(t) \neq 0$. Then an application of Theorem \ref{t1} gives $\eta_k(t) = \hamil^\dag(t)^k \eta$. Noting that $\hamil^\dag(t)^2 = \omega(t)^2 \mathbf{I}_2$, we have
  \[
    \eta_k(t) =
    \begin{cases}
      \omega(t)^k \eta, 	&\text{ when $k$ is even;}\\
      \omega(t)^{k-1} \hamil^\dag(t) \eta, 	&\text{ when $k$ is odd.}
    \end{cases}
  \]
\end{example}

  Next we investigate if a non-Hermitian Hamiltonian is pseudo-Hermitian with respect to a particular $\eta$, then whether or not there exists other Hamiltonians which are pseudo-Hermitian with respect to the same $\eta$. Observe that if an $\eta$ works for two different non-Hermitian Hamiltonians, it still works for their difference, which need not be non-Hermitian in general. We use this observation in the following

\begin{theorem}\label{t2}
  Consider a non-Hermitian Hamiltonian $\hamil$. Suppose $\eta$ is a Hermitian linear automorphism such that $\eta \hamil \eta^{-1} = \hamil^\dag$. Further suppose $K$ is Hermitian and commutes with $\eta$. Let $\Hamil =\hamil + K$. Then $\Hamil$ is also non-Hermitian and pseudo Hermitian with respect to $\eta$. 
\end{theorem}
\begin{proof}
  Clearly, $\Hamil^\dag = \hamil^\dag + K^\dag = \hamil^\dag + K \neq \hamil + K = \Hamil$. So, $\Hamil$ is not Hermitian. Now,
  \begin{align*}
    \eta \Hamil \eta^{-1} &= \eta (\hamil + K) \eta^{-1} \\
    &= \eta \hamil \eta^{-1} + \eta K \eta^{-1}\\
    &= \hamil^\dag + K \eta \eta^{-1}\\
    &= \hamil^\dag + K \\
    &= \Hamil^\dag,
  \end{align*}
  which completes the proof.
\end{proof} 

\begin{remark}\label{quasi-herm}
Recall that a pseudo-Hermitian Hamiltonian is called \emph{quasi-Hermitian} if (\ref{eq:pseudo}) holds with $\eta$ being a \emph{positive} Hermitian linear automorphism\citep{kretschmer2004quasi,scholtz1992quasi,siegl2008quasi}. The positiveness of $\eta$ guarantees that $\eta^{\frac{1}{2}}$ is well-defined and then  $\hamil_{\eta}:=\eta^{\frac{1}{2}}\hamil\eta^{-\frac{1}{2}}$ becomes a Hermitian Hamiltonian. In other words, $\hamil$ becomes a Hermitian Hamiltonian in the induced Hilbert space with induced inner product $\langle  \phi | \psi \rangle_{\eta}:=\langle \phi | \eta \psi \rangle$. Note that Theorem \ref{t2} remains valid for quasi-Hermitian Hamiltonians also.
\end{remark}

\begin{remark}\label{rem:poly}
  Under the hypotheses of Theorem \ref{t2}, $f(K)$ is also Hermitian for any $f(x) \in \mathbb{R}[x]$, the ring of all real polynomials, and commutes with $\eta$. Thus, the same $\eta$ works for $\Hamil = \hamil + f(K)$ also.
\end{remark}

  Below we give two automatic choices for $K$.

\begin{enumerate}
  \item Take $K = \alpha \eta$, where $ \alpha \in \mathbb{R}$. Then $K$ satisfies the hypotheses of the above theorem and hence $\eta$ works for the class of Hamiltonians defined by $\hamil_{\alpha} = \hamil + \alpha \eta$.
  \item Take $K = \alpha \mathbf{I}$, where $\alpha \in \mathbb{R}$. Then clearly $K$ is Hermitian and commutes with every other operator, in particular $\eta$. Thus $K$ satisfies the hypotheses of the above theorem and hence for $\Hamil = \hamil + \alpha \mathbf{I}$, the same $\eta$ works.
\end{enumerate}

  Considering Remark \ref{rem:poly} these two choices can be united into $K=f(\eta)$, where $f(x)\in\mathbb{R}[x]$, and $\eta^0$ is defined to be $\mathbf{I}$.
  
  Thus we see that if one has already constructed the operator $\eta$ for a particular pseudo-Hermitian (or quasi-Hermitian) Hamiltonian $\hamil$, then the same $\eta$ will work for several perturbed versions of $\hamil$. We illustrate this idea with some examples.

\begin{example}
  Consider the Hamiltonian in Example \ref{exm:mostafa}. Here
  \[
    \hamil_{\alpha, k}(t) = \hamil(t) + \alpha \eta_{k}(t),
  \]
  and $\hamil_{\alpha, k}(t)$ is pseudo-Hermitian with respect to $\eta_{k}(t)$.
\end{example}

We shall now consider a class of non-matrix Hamiltonians. It was shown by Bender and Boettcher\citep{bender1998real} that each member of the class of Hamiltonians $\hamil = p^2 + m^2 x^2 - (ix)^N$, $N$ real, has a real spectrum and they conjectured that this is due to $\mathcal{PT}$-symmetry. Although several $\mathcal{PT}$-symmetric Hamiltonians were found to posses real discrete spectra, it was also found that the non-$\mathcal{PT}$-symmetric complex potential of the form \citep{ahmed2001pseudo},
\begin{equation}
\widetilde{V}(x) = \alpha V(x - \beta - i\gamma) \label{potential}
\end{equation}
always yields a real spectrum and Hamiltonians with these potentials are pseudo-Hermitian with respect to
\begin{equation}
\eta = e^{-\theta p},
\end{equation}
where the parameter $\theta=2\gamma$ takes different forms for different potentials. 

It may be noted that $e^{-\theta p}$ has the following two nice properties:
\begin{enumerate}
  \item $e^{-\theta p} c e^{\theta p} = c$, for any constant $c \in \mathbb{C}$, and
  \item $e^{-\theta p} p e^{\theta p} = p$.
\end{enumerate}
  In fact, it follows from the above two properties that $e^{-\theta p} f(p) e^{\theta p} = f(p)$, where $f(x) \in \mathbb{C}[x]$, the ring of all complex polynomials. Now, since $p$ is Hermitian, we can take $K = p$ or more generally, in view of Remark \ref{rem:poly}, 
  \[
    K=f(p),\text{ where }f(x)\in\mathbb{R}[x].
  \]
  So, $\eta = e^{-\theta p}$ renders $\Hamil$ pseudo-Hermitian with an appropriate choice of $\theta$, where
  \[
    \Hamil = p^2 + f(p) + \widetilde{V}(x),
  \]
 and $\widetilde{V}(x)$ is given by (\ref{potential}). It may be noted that the above Hamiltonian represents momentum dependent interaction \citep{swanson2004transition}.

\begin{example}
  In view of the discussion in the above paragraph, let us look at a concrete example. Take $f(x) = \alpha x$ where $\alpha \in \mathbb{R}$. Then
  \[
    \Hamil = p^2 + \alpha p + \widetilde{V}(x).
  \]
\end{example}

\begin{remark}
  The assumption in Theorem \ref{t1} that $\hamil^\dag$ is invertible seems to be a bit restrictive. We can, however, relax this assumption to some extent by combining Theorems \ref{t1} and \ref{t2} in the following manner. Note that invertibility of $\hamil^\dag$ is equivalent to $0$ not being in the spectrum of $\hamil$. Suppose that $\hamil$ has a discrete spectrum $spec(\hamil)$ (By definition, for an operator $\hamil$ on a complex Hilbert space, $spec(\hamil) = \{ \lambda \in \mathbb{C} : \lambda \mathbf{I} - \hamil \text{ is singular} \}$). If $0 \notin spec(\hamil)$, Theorem \ref{t1} is readily applicable. If not, then we can find $\alpha \in \mathbb{R}$, such that $0 \notin spec(\hamil + \alpha \mathbf{I})$. Defining $\Hamil = \hamil + \alpha \mathbf{I}$, it follows that $\Hamil^\dag$ is invertible. Now Theorem \ref{t2} applies to $\hamil = \Hamil + (-\alpha) \mathbf{I}$. So, the class of $\eta$'s generated for $\Hamil$ using Theorem \ref{t1} works for $\hamil$.
\end{remark}

\section{Summary}
Let us now summarize the main results of this paper. We have given sufficient conditions using which one can systematically construct a sequence of $\eta$ operators for a pseudo-Hermitian Hamiltonian starting from a known one. We have also investigated conditions which ensure that a single  $\eta$ operator works for two different pseudo-Hermitian (or quasi-Hermitian) Hamiltonians. As the construction of an $\eta$ operator is not always straightforward for a general non-Hermitian Hamiltonian, these conditions are useful to enlarge the class of non-Hermitian Hamiltonians to which a known $\eta$ operator applies. This ideas are illustrated through a few useful examples.

\bibliography{pseudo}

\begin{thebibliography}{23}%
\makeatletter
\providecommand \@ifxundefined [1]{%
 \@ifx{#1\undefined}
}%
\providecommand \@ifnum [1]{%
 \ifnum #1\expandafter \@firstoftwo
 \else \expandafter \@secondoftwo
 \fi
}%
\providecommand \@ifx [1]{%
 \ifx #1\expandafter \@firstoftwo
 \else \expandafter \@secondoftwo
 \fi
}%
\providecommand \natexlab [1]{#1}%
\providecommand \enquote  [1]{``#1''}%
\providecommand \bibnamefont  [1]{#1}%
\providecommand \bibfnamefont [1]{#1}%
\providecommand \citenamefont [1]{#1}%
\providecommand \href@noop [0]{\@secondoftwo}%
\providecommand \href [0]{\begingroup \@sanitize@url \@href}%
\providecommand \@href[1]{\@@startlink{#1}\@@href}%
\providecommand \@@href[1]{\endgroup#1\@@endlink}%
\providecommand \@sanitize@url [0]{\catcode `\\12\catcode `\$12\catcode
  `\&12\catcode `\#12\catcode `\^12\catcode `\_12\catcode `\%12\relax}%
\providecommand \@@startlink[1]{}%
\providecommand \@@endlink[0]{}%
\providecommand \url  [0]{\begingroup\@sanitize@url \@url }%
\providecommand \@url [1]{\endgroup\@href {#1}{\urlprefix }}%
\providecommand \urlprefix  [0]{URL }%
\providecommand \Eprint [0]{\href }%
\providecommand \doibase [0]{http://dx.doi.org/}%
\providecommand \selectlanguage [0]{\@gobble}%
\providecommand \bibinfo  [0]{\@secondoftwo}%
\providecommand \bibfield  [0]{\@secondoftwo}%
\providecommand \translation [1]{[#1]}%
\providecommand \BibitemOpen [0]{}%
\providecommand \bibitemStop [0]{}%
\providecommand \bibitemNoStop [0]{.\EOS\space}%
\providecommand \EOS [0]{\spacefactor3000\relax}%
\providecommand \BibitemShut  [1]{\csname bibitem#1\endcsname}%
\let\auto@bib@innerbib\@empty
\bibitem [{\citenamefont {Bender}\ and\ \citenamefont
  {Boettcher}(1998)}]{bender1998real}%
  \BibitemOpen
  \bibfield  {author} {\bibinfo {author} {\bibfnamefont {C.~M.}\ \bibnamefont
  {Bender}}\ and\ \bibinfo {author} {\bibfnamefont {S.}~\bibnamefont
  {Boettcher}},\ }\bibfield  {title} {\enquote {\bibinfo {title} {{R}eal
  spectra in non-{H}ermitian {H}amiltonians having
  $\mathcal{P}\mathcal{T}$-symmetry},}\ }\href@noop {} {\bibfield  {journal}
  {\bibinfo  {journal} {Physical Review Letters}\ }\textbf {\bibinfo {volume}
  {80}},\ \bibinfo {pages} {5243} (\bibinfo {year} {1998})}\BibitemShut
  {NoStop}%
\bibitem [{\citenamefont {Bender}, \citenamefont {Boettcher},\ and\
  \citenamefont {Meisinger}(1999)}]{bender1999pt}%
  \BibitemOpen
  \bibfield  {author} {\bibinfo {author} {\bibfnamefont {C.~M.}\ \bibnamefont
  {Bender}}, \bibinfo {author} {\bibfnamefont {S.}~\bibnamefont {Boettcher}}, \
  and\ \bibinfo {author} {\bibfnamefont {P.~N.}\ \bibnamefont {Meisinger}},\
  }\bibfield  {title} {\enquote {\bibinfo {title}
  {$\mathcal{P}\mathcal{T}$-symmetric quantum mechanics},}\ }\href@noop {}
  {\bibfield  {journal} {\bibinfo  {journal} {Journal of Mathematical Physics}\
  }\textbf {\bibinfo {volume} {40}},\ \bibinfo {pages} {2201} (\bibinfo {year}
  {1999})}\BibitemShut {NoStop}%
\bibitem [{\citenamefont {Bender}, \citenamefont {Brody},\ and\ \citenamefont
  {Jones}(2002)}]{bender2002complex}%
  \BibitemOpen
  \bibfield  {author} {\bibinfo {author} {\bibfnamefont {C.~M.}\ \bibnamefont
  {Bender}}, \bibinfo {author} {\bibfnamefont {D.~C.}\ \bibnamefont {Brody}}, \
  and\ \bibinfo {author} {\bibfnamefont {H.~F.}\ \bibnamefont {Jones}},\
  }\bibfield  {title} {\enquote {\bibinfo {title} {Complex extension of quantum
  mechanics},}\ }\href@noop {} {\bibfield  {journal} {\bibinfo  {journal}
  {Physical Review Letters}\ }\textbf {\bibinfo {volume} {89}},\ \bibinfo
  {pages} {270401} (\bibinfo {year} {2002})}\BibitemShut {NoStop}%
\bibitem [{\citenamefont {Fern{\'a}ndez}\ \emph {et~al.}(1998)\citenamefont
  {Fern{\'a}ndez}, \citenamefont {Guardiola}, \citenamefont {Ros},\ and\
  \citenamefont {Znojil}}]{fernandez1998strong}%
  \BibitemOpen
  \bibfield  {author} {\bibinfo {author} {\bibfnamefont {F.}~\bibnamefont
  {Fern{\'a}ndez}}, \bibinfo {author} {\bibfnamefont {R.}~\bibnamefont
  {Guardiola}}, \bibinfo {author} {\bibfnamefont {J.}~\bibnamefont {Ros}}, \
  and\ \bibinfo {author} {\bibfnamefont {M.}~\bibnamefont {Znojil}},\
  }\bibfield  {title} {\enquote {\bibinfo {title} {Strong-coupling expansions
  for the $\mathcal{P}\mathcal{T}$-symmetric oscillators},}\ }\href@noop {}
  {\bibfield  {journal} {\bibinfo  {journal} {Journal of Physics A:
  Mathematical and General}\ }\textbf {\bibinfo {volume} {31}},\ \bibinfo
  {pages} {10105} (\bibinfo {year} {1998})}\BibitemShut {NoStop}%
\bibitem [{\citenamefont {L{\'e}vai}\ and\ \citenamefont
  {Znojil}(2000)}]{levai2000systematic}%
  \BibitemOpen
  \bibfield  {author} {\bibinfo {author} {\bibfnamefont {G.}~\bibnamefont
  {L{\'e}vai}}\ and\ \bibinfo {author} {\bibfnamefont {M.}~\bibnamefont
  {Znojil}},\ }\bibfield  {title} {\enquote {\bibinfo {title} {Systematic
  search for $\mathcal{P}\mathcal{T}$-symmetric potentials with real energy
  spectra},}\ }\href@noop {} {\bibfield  {journal} {\bibinfo  {journal}
  {Journal of Physics A: Mathematical and General}\ }\textbf {\bibinfo {volume}
  {33}},\ \bibinfo {pages} {7165} (\bibinfo {year} {2000})}\BibitemShut
  {NoStop}%
\bibitem [{\citenamefont {L{\'e}vai}, \citenamefont {Cannata},\ and\
  \citenamefont {Ventura}(2002)}]{levai2002pt}%
  \BibitemOpen
  \bibfield  {author} {\bibinfo {author} {\bibfnamefont {G.}~\bibnamefont
  {L{\'e}vai}}, \bibinfo {author} {\bibfnamefont {F.}~\bibnamefont {Cannata}},
  \ and\ \bibinfo {author} {\bibfnamefont {A.}~\bibnamefont {Ventura}},\
  }\bibfield  {title} {\enquote {\bibinfo {title} {$\mathcal{P}\mathcal{T}$
  symmetry breaking and explicit expressions for the pseudo-norm in the {S}carf
  {II} potential},}\ }\href@noop {} {\bibfield  {journal} {\bibinfo  {journal}
  {Physics Letters A}\ }\textbf {\bibinfo {volume} {300}},\ \bibinfo {pages}
  {271--281} (\bibinfo {year} {2002})}\BibitemShut {NoStop}%
\bibitem [{\citenamefont
  {Mostafazadeh}(2002{\natexlab{a}})}]{mostafazadeh2002pseudo1}%
  \BibitemOpen
  \bibfield  {author} {\bibinfo {author} {\bibfnamefont {A.}~\bibnamefont
  {Mostafazadeh}},\ }\bibfield  {title} {\enquote {\bibinfo {title}
  {Pseudo-{H}ermiticity versus $\mathcal{P}\mathcal{T}$-symmetry: the necessary
  condition for the reality of the spectrum of a non-{H}ermitian
  {H}amiltonian},}\ }\href@noop {} {\bibfield  {journal} {\bibinfo  {journal}
  {Journal of Mathematical Physics}\ }\textbf {\bibinfo {volume} {43}},\
  \bibinfo {pages} {205} (\bibinfo {year} {2002}{\natexlab{a}})}\BibitemShut
  {NoStop}%
\bibitem [{\citenamefont
  {Mostafazadeh}(2002{\natexlab{b}})}]{mostafazadeh2002pseudo2}%
  \BibitemOpen
  \bibfield  {author} {\bibinfo {author} {\bibfnamefont {A.}~\bibnamefont
  {Mostafazadeh}},\ }\bibfield  {title} {\enquote {\bibinfo {title}
  {Pseudo-{H}ermiticity versus $\mathcal{P}\mathcal{T}$-symmetry. {II}. {A}
  complete characterization of non-{H}ermitian {H}amiltonians with a real
  spectrum},}\ }\href@noop {} {\bibfield  {journal} {\bibinfo  {journal}
  {Journal of Mathematical Physics}\ }\textbf {\bibinfo {volume} {43}},\
  \bibinfo {pages} {2814} (\bibinfo {year} {2002}{\natexlab{b}})}\BibitemShut
  {NoStop}%
\bibitem [{\citenamefont
  {Mostafazadeh}(2002{\natexlab{c}})}]{mostafazadeh2002pseudo3}%
  \BibitemOpen
  \bibfield  {author} {\bibinfo {author} {\bibfnamefont {A.}~\bibnamefont
  {Mostafazadeh}},\ }\bibfield  {title} {\enquote {\bibinfo {title}
  {Pseudo-{H}ermiticity versus $\mathcal{P}\mathcal{T}$-symmetry {III}:
  {E}quivalence of pseudo-{H}ermiticity and the presence of antilinear
  symmetries},}\ }\href@noop {} {\bibfield  {journal} {\bibinfo  {journal}
  {Journal of Mathematical Physics}\ }\textbf {\bibinfo {volume} {43}},\
  \bibinfo {pages} {3944} (\bibinfo {year} {2002}{\natexlab{c}})}\BibitemShut
  {NoStop}%
\bibitem [{\citenamefont
  {Mostafazadeh}(2002{\natexlab{d}})}]{mostafazadeh2002pseudo4}%
  \BibitemOpen
  \bibfield  {author} {\bibinfo {author} {\bibfnamefont {A.}~\bibnamefont
  {Mostafazadeh}},\ }\bibfield  {title} {\enquote {\bibinfo {title}
  {Pseudo-{H}ermiticity for a class of nondiagonalizable {H}amiltonians},}\
  }\href@noop {} {\bibfield  {journal} {\bibinfo  {journal} {Journal of
  Mathematical Physics}\ }\textbf {\bibinfo {volume} {43}},\ \bibinfo {pages}
  {6343} (\bibinfo {year} {2002}{\natexlab{d}})}\BibitemShut {NoStop}%
\bibitem [{\citenamefont
  {Mostafazadeh}(2002{\natexlab{e}})}]{mostafazadeh2002pseudo5}%
  \BibitemOpen
  \bibfield  {author} {\bibinfo {author} {\bibfnamefont {A.}~\bibnamefont
  {Mostafazadeh}},\ }\bibfield  {title} {\enquote {\bibinfo {title}
  {Pseudo-supersymmetric quantum mechanics and isospectral pseudo-{H}ermitian
  {H}amiltonians},}\ }\href@noop {} {\bibfield  {journal} {\bibinfo  {journal}
  {Nuclear Physics B}\ }\textbf {\bibinfo {volume} {640}},\ \bibinfo {pages}
  {419--434} (\bibinfo {year} {2002}{\natexlab{e}})}\BibitemShut {NoStop}%
\bibitem [{\citenamefont {Longhi}(2009)}]{longhi2009bloch}%
  \BibitemOpen
  \bibfield  {author} {\bibinfo {author} {\bibfnamefont {S.}~\bibnamefont
  {Longhi}},\ }\bibfield  {title} {\enquote {\bibinfo {title} {Bloch
  {O}scillations in {C}omplex {C}rystals with $\mathcal{P}\mathcal{T}$
  symmetry},}\ }\href@noop {} {\bibfield  {journal} {\bibinfo  {journal} {Phys.
  Rev. Lett.}\ }\textbf {\bibinfo {volume} {103}},\ \bibinfo {pages} {123601}
  (\bibinfo {year} {2009})}\BibitemShut {NoStop}%
\bibitem [{\citenamefont {Kottos}(2010)}]{kottos2010optical}%
  \BibitemOpen
  \bibfield  {author} {\bibinfo {author} {\bibfnamefont {T.}~\bibnamefont
  {Kottos}},\ }\bibfield  {title} {\enquote {\bibinfo {title} {Optical physics:
  {B}roken symmetry makes light work},}\ }\href@noop {} {\bibfield  {journal}
  {\bibinfo  {journal} {Nature Physics}\ }\textbf {\bibinfo {volume} {6}},\
  \bibinfo {pages} {166--167} (\bibinfo {year} {2010})}\BibitemShut {NoStop}%
\bibitem [{\citenamefont {Longhi}(2010)}]{longhi2010pt}%
  \BibitemOpen
  \bibfield  {author} {\bibinfo {author} {\bibfnamefont {S.}~\bibnamefont
  {Longhi}},\ }\bibfield  {title} {\enquote {\bibinfo {title}
  {$\mathcal{P}\mathcal{T}$-symmetric laser absorber},}\ }\href@noop {}
  {\bibfield  {journal} {\bibinfo  {journal} {Physical Review A}\ }\textbf
  {\bibinfo {volume} {82}},\ \bibinfo {pages} {031801} (\bibinfo {year}
  {2010})}\BibitemShut {NoStop}%
\bibitem [{\citenamefont {Makris}\ \emph {et~al.}(2008)\citenamefont {Makris},
  \citenamefont {El-Ganainy}, \citenamefont {Christodoulides},\ and\
  \citenamefont {Musslimani}}]{makris2008beam}%
  \BibitemOpen
  \bibfield  {author} {\bibinfo {author} {\bibfnamefont {K.}~\bibnamefont
  {Makris}}, \bibinfo {author} {\bibfnamefont {R.}~\bibnamefont {El-Ganainy}},
  \bibinfo {author} {\bibfnamefont {D.}~\bibnamefont {Christodoulides}}, \ and\
  \bibinfo {author} {\bibfnamefont {Z.~H.}\ \bibnamefont {Musslimani}},\
  }\bibfield  {title} {\enquote {\bibinfo {title} {Beam dynamics in
  $\mathcal{P}\mathcal{T}$ symmetric optical lattices},}\ }\href@noop {}
  {\bibfield  {journal} {\bibinfo  {journal} {Physical Review Letters}\
  }\textbf {\bibinfo {volume} {100}},\ \bibinfo {pages} {103904} (\bibinfo
  {year} {2008})}\BibitemShut {NoStop}%
\bibitem [{\citenamefont {Guo}\ \emph {et~al.}(2009)\citenamefont {Guo},
  \citenamefont {Salamo}, \citenamefont {Duchesne}, \citenamefont {Morandotti},
  \citenamefont {Volatier-Ravat}, \citenamefont {Aimez}, \citenamefont
  {Siviloglou},\ and\ \citenamefont {Christodoulides}}]{guo2009observation}%
  \BibitemOpen
  \bibfield  {author} {\bibinfo {author} {\bibfnamefont {A.}~\bibnamefont
  {Guo}}, \bibinfo {author} {\bibfnamefont {G.}~\bibnamefont {Salamo}},
  \bibinfo {author} {\bibfnamefont {D.}~\bibnamefont {Duchesne}}, \bibinfo
  {author} {\bibfnamefont {R.}~\bibnamefont {Morandotti}}, \bibinfo {author}
  {\bibfnamefont {M.}~\bibnamefont {Volatier-Ravat}}, \bibinfo {author}
  {\bibfnamefont {V.}~\bibnamefont {Aimez}}, \bibinfo {author} {\bibfnamefont
  {G.}~\bibnamefont {Siviloglou}}, \ and\ \bibinfo {author} {\bibfnamefont
  {D.}~\bibnamefont {Christodoulides}},\ }\bibfield  {title} {\enquote
  {\bibinfo {title} {Observation of $\mathcal{P}\mathcal{T}$-symmetry breaking
  in complex optical potentials},}\ }\href@noop {} {\bibfield  {journal}
  {\bibinfo  {journal} {Physical review letters}\ }\textbf {\bibinfo {volume}
  {103}},\ \bibinfo {pages} {93902} (\bibinfo {year} {2009})}\BibitemShut
  {NoStop}%
\bibitem [{\citenamefont {R{\"u}ter}\ \emph {et~al.}(2010)\citenamefont
  {R{\"u}ter}, \citenamefont {Makris}, \citenamefont {El-Ganainy},
  \citenamefont {Christodoulides}, \citenamefont {Segev},\ and\ \citenamefont
  {Kip}}]{ruter2010observation}%
  \BibitemOpen
  \bibfield  {author} {\bibinfo {author} {\bibfnamefont {C.~E.}\ \bibnamefont
  {R{\"u}ter}}, \bibinfo {author} {\bibfnamefont {K.~G.}\ \bibnamefont
  {Makris}}, \bibinfo {author} {\bibfnamefont {R.}~\bibnamefont {El-Ganainy}},
  \bibinfo {author} {\bibfnamefont {D.~N.}\ \bibnamefont {Christodoulides}},
  \bibinfo {author} {\bibfnamefont {M.}~\bibnamefont {Segev}}, \ and\ \bibinfo
  {author} {\bibfnamefont {D.}~\bibnamefont {Kip}},\ }\bibfield  {title}
  {\enquote {\bibinfo {title} {Observation of parity--time symmetry in
  optics},}\ }\href@noop {} {\bibfield  {journal} {\bibinfo  {journal} {Nature
  Physics}\ }\textbf {\bibinfo {volume} {6}},\ \bibinfo {pages} {192--195}
  (\bibinfo {year} {2010})}\BibitemShut {NoStop}%
\bibitem [{\citenamefont {Bender}, \citenamefont {Berry},\ and\ \citenamefont
  {Mandilara}(2002)}]{bender2002generalized}%
  \BibitemOpen
  \bibfield  {author} {\bibinfo {author} {\bibfnamefont {C.~M.}\ \bibnamefont
  {Bender}}, \bibinfo {author} {\bibfnamefont {M.}~\bibnamefont {Berry}}, \
  and\ \bibinfo {author} {\bibfnamefont {A.}~\bibnamefont {Mandilara}},\
  }\bibfield  {title} {\enquote {\bibinfo {title} {Generalized
  $\mathcal{P}\mathcal{T}$ symmetry and real spectra},}\ }\href@noop {}
  {\bibfield  {journal} {\bibinfo  {journal} {Journal of Physics A:
  Mathematical and General}\ }\textbf {\bibinfo {volume} {35}},\ \bibinfo
  {pages} {L467} (\bibinfo {year} {2002})}\BibitemShut {NoStop}%
\bibitem [{\citenamefont {Kretschmer}\ and\ \citenamefont
  {Szymanowski}(2004)}]{kretschmer2004quasi}%
  \BibitemOpen
  \bibfield  {author} {\bibinfo {author} {\bibfnamefont {R.}~\bibnamefont
  {Kretschmer}}\ and\ \bibinfo {author} {\bibfnamefont {L.}~\bibnamefont
  {Szymanowski}},\ }\bibfield  {title} {\enquote {\bibinfo {title}
  {Quasi-{H}ermiticity in infinite-dimensional {H}ilbert spaces},}\ }\href@noop
  {} {\bibfield  {journal} {\bibinfo  {journal} {Physics Letters A}\ }\textbf
  {\bibinfo {volume} {325}},\ \bibinfo {pages} {112--117} (\bibinfo {year}
  {2004})}\BibitemShut {NoStop}%
\bibitem [{\citenamefont {Scholtz}, \citenamefont {Geyer},\ and\ \citenamefont
  {Hahne}(1992)}]{scholtz1992quasi}%
  \BibitemOpen
  \bibfield  {author} {\bibinfo {author} {\bibfnamefont {F.}~\bibnamefont
  {Scholtz}}, \bibinfo {author} {\bibfnamefont {H.}~\bibnamefont {Geyer}}, \
  and\ \bibinfo {author} {\bibfnamefont {F.}~\bibnamefont {Hahne}},\ }\bibfield
   {title} {\enquote {\bibinfo {title} {Quasi-hermitian operators in quantum
  mechanics and the variational principle},}\ }\href@noop {} {\bibfield
  {journal} {\bibinfo  {journal} {Annals of Physics}\ }\textbf {\bibinfo
  {volume} {213}},\ \bibinfo {pages} {74--101} (\bibinfo {year}
  {1992})}\BibitemShut {NoStop}%
\bibitem [{\citenamefont {Siegl}(2008)}]{siegl2008quasi}%
  \BibitemOpen
  \bibfield  {author} {\bibinfo {author} {\bibfnamefont {P.}~\bibnamefont
  {Siegl}},\ }\emph {\bibinfo {title} {Quasi-Hermitian Models}},\ \href@noop {}
  {Ph.D. thesis},\ \bibinfo  {school} {Master’s thesis (Faculty of Nuclear
  Sciences and Physical Engineering, CTU, Prague) ssmf. fjfi. cvut.
  cz/2008/siegl thesis. pdf} (\bibinfo {year} {2008})\BibitemShut {NoStop}%
\bibitem [{\citenamefont {Ahmed}(2001)}]{ahmed2001pseudo}%
  \BibitemOpen
  \bibfield  {author} {\bibinfo {author} {\bibfnamefont {Z.}~\bibnamefont
  {Ahmed}},\ }\bibfield  {title} {\enquote {\bibinfo {title}
  {Pseudo-{H}ermiticity of {H}amiltonians under imaginary shift of the
  coordinate: real spectrum of complex potentials},}\ }\href@noop {} {\bibfield
   {journal} {\bibinfo  {journal} {Physics Letters A}\ }\textbf {\bibinfo
  {volume} {290}},\ \bibinfo {pages} {19--22} (\bibinfo {year}
  {2001})}\BibitemShut {NoStop}%
\bibitem [{\citenamefont {Swanson}(2004)}]{swanson2004transition}%
  \BibitemOpen
  \bibfield  {author} {\bibinfo {author} {\bibfnamefont {M.~S.}\ \bibnamefont
  {Swanson}},\ }\bibfield  {title} {\enquote {\bibinfo {title} {Transition
  elements for a non-{H}ermitian quadratic {H}amiltonian},}\ }\href@noop {}
  {\bibfield  {journal} {\bibinfo  {journal} {Journal of Mathematical Physics}\
  }\textbf {\bibinfo {volume} {45}},\ \bibinfo {pages} {585} (\bibinfo {year}
  {2004})}\BibitemShut {NoStop}%
\end{thebibliography}%
\end{document}